\documentclass[nocopyrightspace]{sigplanconf}

\usepackage{xspace}
\usepackage{amssymb}
\usepackage{graphicx}
\usepackage{amsmath}

\usepackage{amsthm}
\newcommand{\AtomicUpdates}{\ensuremath{{\cal U}}}

\newcommand{\XP}{\ensuremath{\mathit{XP}}}

\newcommand{\semantics}[1]{\ensuremath{[\![ #1 ]\!]}}
\newcommand{\policy}{\ensuremath{\mathcal{P}}}

\newcommand{\powerset}{\mathcal{P}}

\newcommand{\ktext}{\ensuremath{\mathsf{text}}}

\newcommand{\coNP}{\ensuremath{\textsc{coNP}}}
\newcommand{\PTIME}{\ensuremath{\textsc{ptime}}}
\newcommand{\EXPTIME}{\ensuremath{\textsc{exptime}}}

\newcommand{\locpath}[2]{\mathit{#1}\:{:}{:}\:{#2}}
\newcommand{\udelete}[1]{\ensuremath{\mbox{\it delete}(#1)}}
\newcommand{\uupdate}[2]{\ensuremath{\mbox{\it update}(#1,#2)}}
\newcommand{\uinsert}[2]{\ensuremath{\mbox{\it insert}(#1,#2)}}
\newcommand{\Approx}{\ensuremath{\mathsf{Approx}}}
\newcommand{\sqleq}{\sqsubseteq}

\newcommand{\allowed}{\ensuremath{\mathcal{A}}}
\newcommand{\denied}{\ensuremath{\mathcal{D}}}
\newcommand{\ds}{\ensuremath{\mathsf{ds}}}
\newcommand{\crp}{\ensuremath{\mathsf{cr}}}
\newcommand{\rng}{\ensuremath{\mathrm{rng}}}

\newcommand{\MTree}{\ensuremath{\mathsf{MTree}}}

\newcommand{\eval}[2]{\ensuremath{\semantics{#1}{(#2)}}}

\newcommand{\comp}{\circ}

\newcommand{\ktrue}{\ensuremath{\mathsf{true}}\xspace}

\newcommand{\kand}{\ensuremath{\mathrel{\mathsf{and}}}\xspace}

\newcommand{\kchild}{\ensuremath{\mathsf{child}}\xspace}
\newcommand{\kattribute}{\ensuremath{\mathsf{attribute}}\xspace}
\newcommand{\kself}{\ensuremath{\mathsf{self}}\xspace}
\newcommand{\kdescendant}{\ensuremath{\mathsf{descendant}}\xspace}

\newcommand{\XPFull}{\ensuremath{\XP^{(/,//,*,[\:],=,@)}}}

\newcommand{\XPCDSF}{\ensuremath{\XP^{(/,//,*,[\:])}}}
\newcommand{\XPChild}{\ensuremath{\XP^{(/)}}}

\newcommand{\XPCDS}{\ensuremath{\XP^{(/,//,*)}}}

\newcommand{\XPCF}{\ensuremath{\XP^{(/,[\:])}}}
\newcommand{\XPCSF}{\ensuremath{\XP^{(/,*,[\:])}}}

\newcommand{\labels}{\mathsf{labels}}
\newcommand{\labelsQ}{\mathsf{labels^Q}}

\newcommand{\LPaths}{\mathsf{LP}}
\newcommand{\FPaths}{\mathsf{FP}}
\newcommand{\FPathsQ}{\mathsf{FP^Q}}
\newcommand{\AB}[1]{\langle\!\langle#1\rangle\!\rangle}

\newtheorem{proposition}{Proposition}
\newtheorem{corollary}{Corollary}
\newtheorem{theorem}{Theorem}
\newtheorem{lemma}{Lemma}
\theoremstyle{definition}
\newtheorem{example}{Example}
\newtheorem{definition}{Definition}
\theoremstyle{remark}

\begin{document}

\title{Static Enforceability of XPath-Based Access Control Policies}
\authorinfo{James Cheney}{
University of Edinburgh}{
jcheney@inf.ed.ac.uk}


\maketitle

\begin{abstract}
We consider the problem of extending XML databases with fine-grained,
high-level access control policies specified using XPath expressions.
Most prior work checks individual updates \emph{dynamically}, which is
expensive (requiring worst-case execution time proportional to the
size of the database).  On the other hand, \emph{static} enforcement
can be performed without accessing the database but may be incomplete,
in the sense that it may forbid accesses that dynamic enforcement
would allow.  We introduce topological characterizations of XPath
fragments in order to study the problem of determining when an access
control policy can be enforced statically without loss of precision.
We introduce the notion of \emph{fair} policies that are statically
enforceable, and study the complexity of determining fairness and of
static enforcement itself.




\end{abstract}

\section{Introduction}\label{sec:introduction}

Access control policies for XML documents or databases have been
studied extensively over the past 10
years~\cite{luo04cikm,cho02vldb,qi05cikm,damiani02tissec,bertino02tissec,stoica02ifip,fan04sigmod,yu04tods,murata06tissec,ayyagari07ccs,kuper09ijis}.
Most of this work focuses on high-level, declarative policies based on
XPath expressions or annotated schemas; declarative policies are
considered easier to maintain and analyze for vulnerabilities than the
obvious alternative of storing ad hoc access control annotations
directly in the database itself~\cite{fisler05icse}.  However, this
convenience comes at a cost: enforcing fine-grained, rule-based
policies can be expensive, especially for updates.  In this paper we
consider the problem of efficient enforcement of access control
policies involving update operations, where permissions are specified
using downward monotone XPath access control rules.

An example of such a policy, specifying the allowed and forbidden updates
for nurses in a hospital database, is shown in
Figure~\ref{fig:policy}.  The policy is parameterized by data values
$\$wn$ (ward number) and $\$uid$ (user id); these values are available
as part of the request so can be treated as constants.

The first three positive rules specify that nurses may insert data
into any patient records, may update information about patients in
their own ward, and may update their own phone number; the last two
negative rules specify that nurses may not insert or update treatment
elements.  Some sample data is shown in Figure~\ref{fig:example}.

Most prior work on XML access control focuses on controlling
read-access, and access control for read-only XML data is now
well-understood.  Some techniques, such as
filtering~\cite{cho02vldb,luo04cikm} and security
views~\cite{stoica02ifip,fan04sigmod,kuper09ijis}, hide sensitive data
by rewriting queries or providing sanitized views.  Other access
control techniques rely for efficiency on auxiliary data structures
(such as access control annotations~\cite{koromilas09sdm}, or
``compressed accessibility maps''~\cite{yu04tods}).  Static analysis
has been proposed to avoid dynamic checks~\cite{murata06tissec} or speed
reannotation~\cite{koromilas09sdm}.

However, access control for updates still poses challenges that previous
work on read-only access does not fully address, and XML databases
still typically lack support for fine-grained access control.  Prior
work~\cite{ayyagari07ccs,koromilas09sdm} suggests two obvious dynamic approaches to enforcement of
write-access control policies: \emph{query-based enforcement},
analogous to filtering, in which we use the policy rules and update
request to generate Boolean queries that answer ``true'' if the update
is allowed and ``false'' if not, and \emph{annotation-based
  enforcement}, in which the rules are used to place annotations on
the data indicating which updates are allowed on each node.
In annotation-based enforcement, when an update is performed the
annotations need to be updated to restore consistency with the policy;
query-based enforcement has no such maintenance overhead.

\begin{figure}[tb]
\begin{small}
 $\mbox{\it Nurse\/}(\mbox{\${\it wn\/}},\mbox{\${\it uid\/}})$:\\
  $
  \begin{array}{ll}
    R_1:&+\uinsert{//\mathtt{patient//*}}{*}\\
    R_2:&+\uupdate{//\mathtt{patient[@\mathtt{wardNo}=\$\mathit{ wn}]/*}}{*}\\
    R_3:&+\uupdate{\mathtt{//nurse[@id=\mathit{\$uid}]/phone/*}}{\ktext()}\\
    R_4:&-\uinsert{//*}{\mathtt{treatment}}\\
    R_5:&-\uupdate{\mathtt{//treatment}}{*}\\
  \end{array}
  $
 \end{small}
\caption{Policy example}
 \label{fig:policy}
\includegraphics[scale=0.25]{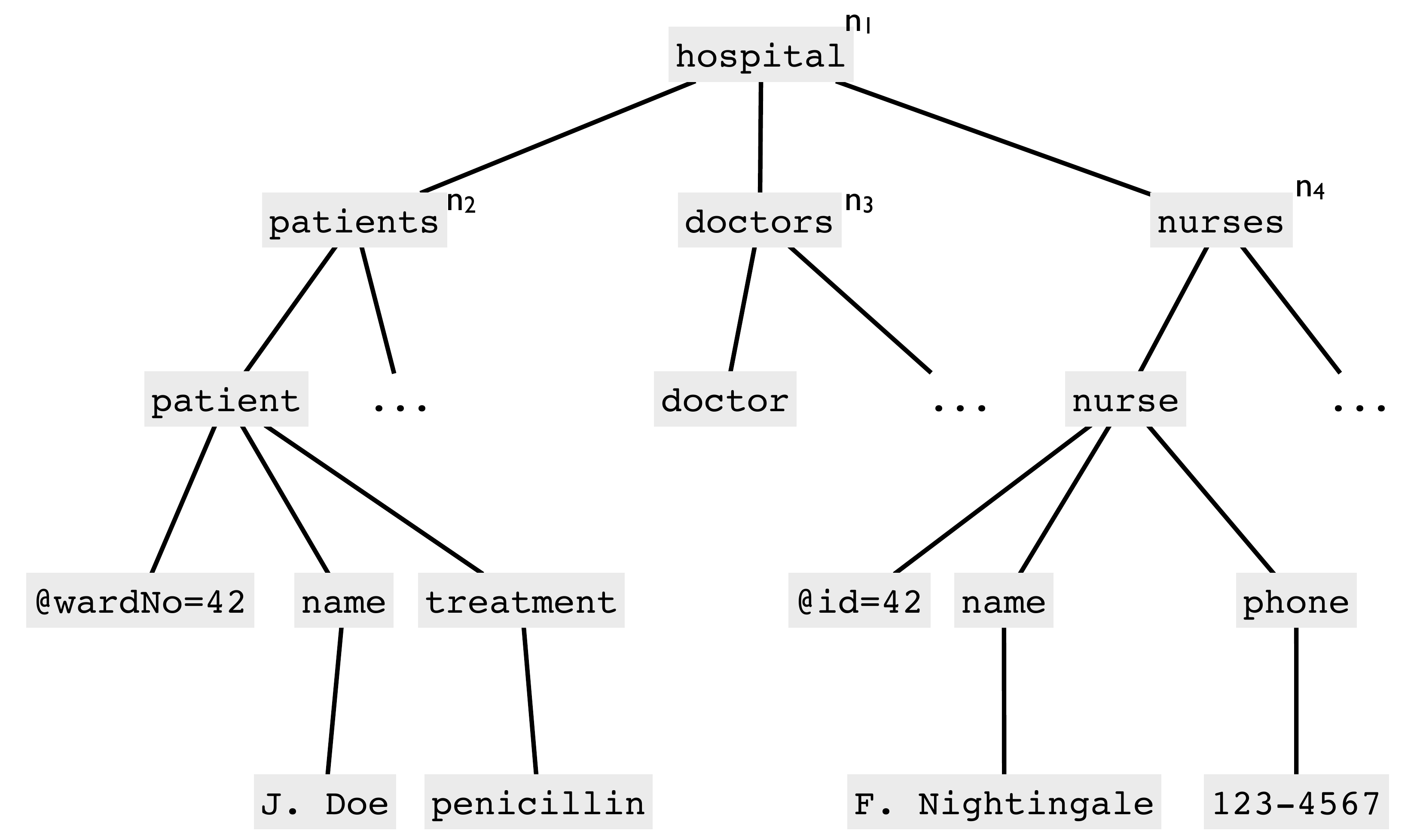}
\caption{Example data}
\label{fig:example}
\end{figure}

To illustrate, consider the data tree in Fig.~\ref{fig:example}.
Suppose nurse $n_{123}$ wishes to insert a new \verb|patient| record
represented by an XML tree $T$.  A client-side program issues an
XQuery Update expression \verb|insert T into /hospital/patients|.
Executing this update yields an atomic update $\uinsert{n_2}{T}$ where
$n_2$ is the node id of the $/\mathtt{hospital}/\mathtt{patients}$ node.  This update is
allowed dynamically by the policy, and this can be checked by
executing a query against the database to select those nodes where
patient insertion is allowed, or by maintaining annotations that
encode this information for all operations.  

Since XPath evaluation is in polynomial time (in terms of data
complexity)~\cite{GottlobKPS05}, both query-based and annotation-based
approaches are tractable in theory, but can be expensive for large
databases.  Koromilas et al.~\cite{koromilas09sdm} found that checking
whether an update is allowed is much faster using annotations than
using queries, but even with static optimizations, the overhead of
maintaining the annotations can still be prohibitively expensive for
large databases.  Both approaches can in the worst case require a complete
traversal of the database; in practice, Koromilas et
al.~\cite{koromilas09sdm} found that incremental maintenance of
annotation-based enforcement requires a few seconds per update even
for databases of modest size.

This strongly motivates an alternative approach that avoids any
dependence on the actual data: static analysis of the rules and
updates to check whether a proposed update is
allowed~\cite{murata06tissec}.  This approach draws upon exact static
analysis algorithms for intersection~\cite{hammerschmidt05ideas} and
containment~\cite{miklau04jacm} of downward XPath.  Intersection is
decidable in polynomial time, but containment for expressive fragments
of XPath can be intractable in the size of the path expressions
involved; even so, for a fixed policy such tests could still be much
faster than dynamic enforcement, because they depend only on the
policy and update size, not that of the data.

To illustrate via our running example, instead of checking the actual
atomic update against the actual data, we can consider a static
approach, under the assumption that the database does not allow atomic
updates directly but instead only accepts updates
specified using a high-level update language such as XQuery
Update~\cite{xquery-update}. For example, the user-provided update $u$
could be
\begin{verbatim}
insert T into /hospital/patients
\end{verbatim}
In prior work, we have introduced static analyses that provide a
conservative static approximation of the possible effects of an
update~\cite{benedikt09dbpl}.  We call such representations
\emph{update capabilities}.   In our
approach, the system first approximates $u$ via an update
  capability
\[U = \uinsert{/\mathtt{hospital}/\mathtt{patients}}{\mathtt{patient}}\;\]
Here, the second argument $\mathtt{patient}$ indicates the type of
node being inserted, that is, the root label of $T$. Again, in this
case the access is allowed, since $U$ is contained in the positive
rule $R_1$ and does not overlap with any of the negative rules
$R_4,R_5$.

However, purely static enforcement may not give the same results as
dynamic enforcement: put another way, for some policies and updates,
it may be impossible to statically determine whether the update is
allowed.  Static enforcement would either deny access in such a case
or fall back on dynamic techniques.  We call a policy \emph{fair} when
this is not the case: that is, when purely static and dynamic
enforcement coincide.  

For example, if we add a rule
$-\udelete{\mathtt{//patient[treatment]}}$ to the example policy in
Figure~\ref{fig:policy}, the resulting policy is unfair with respect
to any monotone fragment of XPath, because there is no way to specify
a static update request that guarantees the absence of a $\mathtt{treatment}$
child in the updated $\mathtt{patient}$ subtree.  Fair policies are of interest
because they can be enforced statically, avoiding any dependence on
the size of the data.

In this paper we consider the \emph{fairness} problem: \emph{given a
  policy language and a policy in that language, determine whether the
  policy is statically enforceable}.  We focus on subsets of downward,
unordered, monotone XPath.  In this context, \emph{downward} and
\emph{unordered} refers to the fact that we consider only the self,
child and descendant axes that navigate downward into the tree and are
insensitive to order (though our results also apply to ordered trees),
and \emph{monotone} refers to the fact that we exclude features such
as negative path tests or difference operations, so that all of the
XPath expressions we consider have monotone semantics.  We use
notation $\XP^{(S)}$, where $S$ is a set of XPath features such as
child ($/$), descendant $(//)$, filter $([~])$ or wildcard ($*$) to
denote different fragments of downward XPath.

Our key insight is based on a shift of perspective. A conventional
view of the semantics of an XPath expression $p$ over a given tree $T$
is as a set of selected nodes $n$ obtained by evaluating $p$ from the
root of $T$.  Instead, we consider the semantics of $p$ to be the set
of pairs $(T,n)$.  We consider the topological spaces generated by
different fragments of XPath.  A policy is fair (with respect to
updates specified in a given fragment $\XP$) if and only if its
semantics denotes an open set in the topology generated by $\XP$.
Intuitively, the reason for this is that a policy is fair if any
update dynamically allowed by the policy is contained in a statically
allowed update capability.  The atomic updates are points of the
topological space, the update capabilities denote basic open sets.

Based on this insight, we first prove that fairness is monotonic in the
fragment $\XP$ used for updates: that is, making the XPath
characterizations of updates more precise never damages fairness.
Second, we show that all policies over $\XPCDS$ are fair with respect
to $\XPChild$ (or any larger fragment). We show that it is \coNP-complete
to decide whether a policy over $\XPCDSF$ is fair with respect to
$\XPCF$; however, policies that only use filters in positive
rules are always fair.  We show that for update operations with a
bounded number of descendant steps, static enforcement is decidable in
polynomial time. We sketch how these results can be extended to handle
policies with attributes and data value tests.

The structure of the rest of this paper is as follows: In
Section~\ref{sec:prelim} we review the model of write-access control
policies introduced in prior work.  We define fairness and give its
topological characterization in Section~\ref{sec:fair} and present the
main results in
Section~\ref{sec:results}. Section~\ref{sec:discussion} discusses the
implications of our results and generalizations.  We conclude with
discussions of related and future work in Sections~\ref{sec:related}
and~\ref{sec:concl}.



\section{Preliminaries}
\label{sec:prelim}

\paragraph*{XML trees} We model XML documents as unordered, unranked
trees. Let $\Sigma$ be an \emph{element name alphabet}, $\Gamma$ an
\emph{attribute name alphabet}, and $D$ a \emph{data domain}.  We
assume that $\Sigma$, $\Gamma$, and $D$ are infinite and mutually
disjoint.  We consider an XML document to be a tree $T = (V_T,
E_T,R_T, \lambda_T)$, where $\lambda_T : V_T \to \Sigma\uplus (\Gamma
\times D) \uplus D$ is a function mapping each node to an appropriate
label, $E_T \subseteq V_T \times V_T$ is the edge relation, and $R_T$
is a distinguished node in $V_T$, called the root node.  We
distinguish between \emph{element nodes} labeled with $l \in \Sigma$,
\emph{attribute nodes} labeled with attribute-value pairs $(@f, d)\in
\Gamma \times D$, and \emph{data nodes} labeled with elements of $d
\in D$; attribute and data nodes must be leaves. We do not assume that
an XML DTD or schema is present.

\paragraph*{XPath}~The fragment of downward XPath used in update
operations and policies is defined as follows:
\[\begin{array}{lrcl}
\text{Paths}  & p & ::= & \alpha::\phi\mid p / p' \mid p[q]\\
 \text{Filters} & q & ::= & p \mid q \kand q \mid @f = d \mid \ktrue \\
 \text{Axes} & \alpha& ::= & \kself \mid \kchild \mid \kdescendant \mid \kattribute\\
 \text{Node tests} &\phi& ::= & l \mid  * \mid f \mid \ktext()
\end{array}\]%
Absolute paths are written $/p$; we often omit the leading slash when
this is obvious from context.  
Here, $l$ is an element label from $\Sigma$, $f$ is an attribute name
from $\Gamma$, and $d$ is a data value or parameter name.  Wildcard
$*$ matches any element or text node.
The expressions are
built using only the child, descendant and attribute axes of XPath and conditions
that test for the existence of paths or constant values of
attributes.  We use the standard abbreviated forms of XPath
expressions in examples.  For example, $/a//b[*/@d]$ abbreviates $/\locpath{\kchild}{a}/\locpath{\kdescendant}{b}[\locpath{\kchild}{*}/\locpath{\kattribute}{d}]$.  We
write $\eval{p}{T}$ for the set of nodes of a tree $T$ obtained from
evaluating XPath expression $p$ on the root node of $T$. We also write
$\semantics{\phi}$ for the subset of node labels $\Sigma \uplus (\Gamma \times D)
\uplus D$ matching $\phi$.  These semantics are defined in Figure~\ref{fig:semantics},
following standard
treatments~\cite{Benedikt03,GottlobKPS05,wadler00:_two_seman_xpath}. 

\begin{figure}[tb]
  \[
\begin{array}{rcl}
\semantics{\phi} &\subseteq & \Sigma \uplus (\Gamma \times D) \uplus D\medskip\\
\semantics{*} &=& \Sigma\\
\semantics{l} &=& \{l\}\\
\semantics{f} &=& \{(f,d) \mid d \in D\}\\
\semantics{\ktext()} &=& D\medskip\\
 A\eval{\alpha}{T}  &\subseteq& V_T \times V_T\\
 A\eval{\kself}{T} &=& \{(x,x) \mid x \in V_T\}\\
 A\eval{\kchild}{T} &=& E_T\\
 A\eval{\kdescendant}{T} &=& E_T^+\\
 A\eval{\kattribute}{T} &=& \{(m,n) \mid \lambda_T(n) = (@f,d)\}\medskip\\
 P\eval{p}{T} &\subseteq& V_T \times V_T\\
 P\eval{\alpha::\phi}{T} &=& \{(v,w) \in A\eval{\alpha}{T} \mid  \lambda_T(w) \in \semantics{\phi}\}\\
 P\eval{p[q]}{T}&=& \{(v,w) \in P\eval{p}{T} \mid w \in  Q\eval{q}{T}\}\\
P\eval{p/p'}{T} &=& \{(v,w) \mid \exists x \in V_T.  (v,x)  \in P\eval{p}{T} ,\\
&&\qquad \qquad
(x,w) \in P\eval{p'}{T}\}\smallskip\\
 Q\eval{q}{T} &\subseteq& V_T\\
 Q\eval{p}{T} &=& \{v \mid \exists w \in V_T. (v,w) \in P\eval{p}{T}\}\\
 Q\eval{q \kand q'}{T} &=& Q\eval{q}{T} \cap Q\eval{q'}{T}\\
 Q\eval{@f=d}{T} &=& \{v \mid \exists w. (v,w) \in E_T, \lambda_T(w) = (@f,d)\}\\
Q\eval{\ktrue}{T} &=& V_T
\medskip\\
 \eval{p}{T} &\subseteq& V_T\\
 \eval{p}{T} &=& \{v \mid (R_T,v) \in P\eval{p}{T}\}
\end{array}\]
  \caption{Semantics of \XPFull}
  \label{fig:semantics}
\end{figure}

 We write $\XP^{(S)}$, for
$S \subseteq \{/,*,//,[\:],=, @\}$, for the sublanguage of the
above XPath expressions that includes the features in $S$.  For
example, $\XP^{(/,//,=,@)}$ includes $/a/b[@c=\mbox{"foo"}]$, but not
$//a/*$.


We say that an XPath expression $p$ is \emph{contained in} another
expression $p'$ (written $p \sqsubseteq p'$) if for every XML tree $T$,
$\eval{p}{T} \subseteq \eval{p'}{T}$.  We say that two XPath
expressions are \emph{disjoint} if their
intersection is empty: that is, for every $T$, $\eval{p}{T} \cap
\eval{p'}{T} = \emptyset$.  Otherwise, we say $p$ and $p'$
\emph{overlap}.

As for relational queries, containment and satisfiability are closely
related for XPath queries, and both problems have been studied for
many different fragments of XPath.  Containment has been studied for
downward XPath expressions ($\XPCDSF$) by Miklau and
Suciu~\cite{miklau04jacm} and for larger fragments by
others~\cite{neven06lmcs,benedikt08jacm,tencate09jacm}.  Specifically, Miklau and
Suciu showed that containment is \coNP-complete for $\XPCDSF$ and
presented a complete, exponential algorithm and an incomplete,
polynomial time algorithm, which is complete in restricted cases.
Polynomial algorithms for testing overlap of XPath expressions in the
fragment $\XPCDSF$ have been studied in~\cite{hammerschmidt05ideas};
however, both satisfiability and containment for XPath with child
axis, filters and negation is PSPACE-hard~\cite{benedikt08jacm}, and
the complexity of containment increases to EXPTIME-hard when the
descendant axis is added.  Containment for XPath 2.0, which includes
negation, equality, quantification, intersection, and difference
operations, rapidly increases to $\EXPTIME$ or non-elementary
complexity~\cite{tencate09jacm}.

\paragraph*{Atomic Updates} We consider atomic updates of the form: 
\[u ::= \uinsert{n}{T'} \mid \uupdate{n}{T'} \mid\udelete{n}\]
where $n$ is a node expression, and $T'$ 
is an XML tree. An insert operation \uinsert{n}{T'} is
applied to a tree $T$ by adding a copy of $T'$ as a child of node $n$
(recall that we consider unordered trees so the order does not matter). The operation \udelete{n} deletes the
subtree of $n$, and likewise the operation $\uupdate{n}{T'}$ replaces
the selected node with $T'$.  We write
$\AtomicUpdates(T)$ for the set of all atomic updates applicable to
the nodes of $T$.
We omit a definition of the semantics of atomic updates on trees,
since it is not necessary for the results of the paper.


\paragraph*{Update Capabilities}~We consider update capabilities of the form
\[U ::= \uinsert{p}{\phi} \mid \uupdate{p}{\phi} \mid \udelete{p} \]
where $p$ is an XPath expression, and $\phi$ is a node test constraining
the tree that can be inserted.  Intuitively, an update capability
describes a set of atomic update operations that a user is allowed
or forbidden to perform in the context of a given policy.  An update
capability is interpreted (with respect to a given tree) as defining a
set of atomic updates:
\begin{eqnarray*}
\semantics{\uinsert{p}{\phi}}(T) &=& \{\uinsert{n}{T'} \mid n \in
\eval{p}{T}, \\&&\qquad
\lambda_{T'}(R_{T'})\in \semantics{\phi}\}\\
\semantics{\uupdate{p}{\phi}}(T) &=& \{\uupdate{n}{T'} \mid n \in
\eval{p}{T},  \\&&\qquad\lambda_{T'}(R_{T'})\in \semantics{\phi}\}\\
\semantics{\udelete{p}}(T) &=& \{\udelete{n} \mid n \in \eval{p}{T}\}
\end{eqnarray*}

%

\paragraph*{Access Control Policies}~ Following prior work
(e.g.~\cite{fundulaki07sacmat,koromilas09sdm}), we define {\em access
  control policies} $\policy=(\ds, \crp,\allowed,\denied ) $ with four
components: a {\em default semantics} $\ds \in \{+,-\}$, a {\em
  conflict resolution policy} $\crp \in \{+,-\}$, and sets $\allowed$
and $\denied$ of allowed and denied capabilities, described by XPath
expressions.  The default semantics indicates whether an operation is
allowed if no rules are applicable.  The conflict resolution policy
resolves conflicts when an operation matches both a positive rule and
a negative rule.  The semantics $\semantics{\policy}$ of a policy
$\policy = (\ds, \crp,\allowed, \denied)$ is given in
Figure~\ref{fig:policy_semantics}, defined as a function from trees
$T$ to sets of allowed atomic updates $\semantics{\policy}(T)$q.  For
example, in the deny--deny case, the accessible nodes are those for
which there is a capability granting access and no capabilities
denying access.  Note that the allow--deny and deny--allow cases are
degenerate cases of the other two when $\allowed =\emptyset$ or
$\denied=\emptyset$ respectively.

\begin{figure}[tb]
\begin{eqnarray*}
\semantics{( +, +,\allowed, \denied)}(T) & = & \AtomicUpdates(T) - (\semantics{\denied}(T) - \semantics{\allowed}(T))\\
\semantics{(-, +,\allowed, \denied )}(T) & = & \semantics{\allowed}(T)\\
\semantics{(+, -,\allowed, \denied )}(T) & = & \AtomicUpdates(T) -\semantics{\denied}(T)\\
\semantics{(-, -,\allowed, \denied )}(T) & = & \semantics{\allowed} (T)- \semantics{\denied}(T)
\end{eqnarray*}
\caption{Semantics of access control policies as the set of allowed
  atomic updates}
\label{fig:policy_semantics}
\end{figure}

\paragraph*{Enforcement Models} We now define the two enforcement
models: \emph{dynamic} and \emph{static}.  
\begin{definition}
  An update $u$ is \emph{(dynamically) allowed} on tree $T$ if
  $\semantics{u}(T) \in \semantics{\policy}(T)$.  An update capability
  $U$ is \emph{statically allowed} provided that for all $T'$, we have
  $\semantics{U}(T') \subseteq \semantics{\policy}(T')$.
\end{definition}
For any policy, if $u \in \semantics{U}(T) \subseteq
\semantics{\policy}(T)$, then clearly $u$ is dynamically allowed on
$T$.  The reverse is not necessarily the case, depending on the policy
and class $\XP$ of paths used in update capabilities.  
\begin{definition}
  A policy $\policy$ is \emph{fair} with respect to XPath fragment
  $\XP$ provided that whenever $\policy$ allows $u$ on $T$, there
  exists $U$ expressible in $\XP$ such that $u \in \semantics{U}(T)$
  and $\policy$ statically allows $U$.
\end{definition}
\begin{example}\label{ex:unfair}
  Fairness depends critically upon the class of paths that may be used
  to specify updates.  If we consider updates with respect to
  $\XPCDSF$, an example of an unfair policy is $\policy =
  (-,-,\{\udelete{/a}\},\{\udelete{/a[b]}\})$.  Static enforcement cannot
  ever allow a deletion at $/a$ because there is no way (within
  $\XPCDSF$) to specify an update that only applies to nodes
  that have no $b$ child.  Fairness could be recovered by increasing
  the expressive power of updates, for example to allow negation in
  filters; however, this makes checking containment considerably more
  difficult~\cite{neven06lmcs,benedikt08jacm,tencate09jacm}. On the other hand, constraints such as attribute
  uniqueness mean that some policies with filters in negative rules are fair:
  for example, $(-,-,\{\udelete{/a[@b=c]}\}, \{\udelete{/a[@b=d]}\})$
  is fair.
\end{example}




\section{Topological characterization of fairness}
\label{sec:fair}

For simplicity, we initially limit attention to $\XPCDSF$ and 
update capabilities and policies involving only delete capabilities, and
abuse notation by identifying $\udelete{p}$ with $p$, and thinking of
$\allowed$ and $\denied$ as sets of paths.  We adopt an alternative view of the
semantics of paths and policies.  Let $\MTree$ be the set of pairs
$(T,n)$ where $n \in V_T$.  Such pairs are called \emph{marked trees};
they are essentially \emph{tree patterns} (or \emph{twig queries}) for
XPath expressions in $\XPCF$.  We sometimes also consider doubly
marked trees, that is, structures
$(T,n,m)$ with two marked nodes.  XPath expressions and policies can be
interpreted as sets of (singly or doubly) marked trees:
\begin{definition}
  We define the marked tree semantics of absolute paths $\AB{p}$,
  paths $P\AB{p}$, and qualifiers $Q\AB{q}$ as shown in
  Figure~\ref{fig:reformulated-semantics}.  The marked tree semantics
  of  policies is defined as $\AB{\policy} =
  \{(T,n) \mid \udelete{n} \in \eval{\policy}{T}\}$.
\end{definition}
The following lemma summarizes the relationship between the original
and reformulated semantics:
\begin{lemma}
\begin{enumerate}
\item $\AB{p} = \{(T,n) \mid n \in \eval{p}{T}\}$
\item $p \sqsubseteq p'$ if and only if $\AB{p} \subseteq
\AB{p'}$
\item $p$ overlaps $p'$ if and only if $\AB{p} \cap
\AB{p'} \neq \emptyset$.   
  \end{enumerate}
\end{lemma}

\begin{figure}[tb]
  \[
\begin{array}{rcl}
 \AB{p} &=& \{(T,v) \mid (T,R_T,v) \in P\AB{p}\} \smallskip\\
  P\AB{\alpha::\phi} &=& \{(T,v,w) \mid (v,w)  \in
  A\eval{\alpha}{T}\}\\
P\AB{p/p'} &=& \{(T,v,w) \mid \exists x \in V_T. (T,v,x)  \in
P\AB{p},\\
&&\qquad 
(T,x,w) \in P\AB{p'}\}\\
 P\AB{p[q]} &=& \{(T,v,w) \in P\AB{p} \mid (T,w) \in Q\AB{q}\}\smallskip\\
 Q\AB{p} &=& \{(T,v) \mid \exists w \in V_T. (T,v,w) \in P\AB{p}\}\\
 Q\AB{q \kand q'} &=& Q\AB{q} \cap Q\AB{q'}\\
 Q\AB{@f=d} &=& \{(T,v) \mid \exists w. (v,w) \in E_T, \\
&&\qquad \lambda_T(w) = (@f,d)\}
\end{array}\]
  \caption{Reformulated semantics of \XPFull}
  \label{fig:reformulated-semantics}
\end{figure}

Moreover, we define the \emph{$\XP$-underapproximation} of a set $S
\subseteq \MTree$ as $\Approx_{\XP}(S) = \bigcup\{\AB{p} \mid p \in
\XP, \AB{p} \subseteq S\}$.  Fairness can be reformulated directly in
terms of the underapproximation operation:

\begin{proposition}\label{prop:fair}
A policy $\policy$ is fair with respect to $\XP$ if and only if 
$\AB{\policy} = \Approx_{\XP}(\AB{\policy})$.
\end{proposition}
\begin{proof}
  For the forward direction, suppose $\policy$ is fair.  First note
  that $\AB{\policy} \supseteq \Approx_{\XP}(\AB{\policy})$ holds for
  any policy since the right-hand side is a union of sets contained in
  $\AB{\policy}$.  Suppose that $(T,n) \in \AB{\policy}$.  Then since
  $\policy$ is fair, there exists a $p$ such that $(T,n) \in \AB{p}$
  and $\AB{p} \subseteq \AB{\policy}$.  This implies that $(T,n) \in
   \Approx_{\XP}(\AB{\policy})$.

  For the reverse direction, suppose $\AB{\policy} =
  \Approx_{\XP}(\AB{\policy})$ and suppose
  $\udelete{n}$ is allowed on $T$, that is, $(T,n) \in \AB{\policy}$.
  Then there must exist some $p \in \XP$ such that $(T,n)
  \in \AB{p}$ and $\AB{p} \subseteq \AB{\policy}$.  This implies that
  $p$ is statically allowed, as required for fairness. 
\end{proof}

Fairness is obviously preserved by moving to a
larger XPath fragment $\XP'$:
\begin{corollary}\label{cor:mono}
  If $\XP \subseteq \XP'$, then if $\policy$ is fair with respect to $\XP$
  then it is also fair with respect to $\XP'$.
\end{corollary}

Recall that a topological space is a structure $(X,\tau)$ where
$\tau\subseteq \powerset(X)$ is a collection of \emph{open sets} that
contains $\emptyset$ and $X$, and is closed under finite intersections
and arbitrary unions.  The complement of an open set
is called \emph{closed}.  A \emph{basis} $\mathcal{B}$ for $X$ is
collection of subsets of $X$ such that $\bigcup \mathcal{B} = X$ and
whenever $x \in B_1 \cap B_2$, there exists $B \in \mathcal{B}$ such
that $x \in B \subseteq B_1 \cap B_2$.  A basis $\mathcal{B}$ for $X$
gives rise to a topology $\tau_\mathcal{B}$ for $X$, formed by closing
$\mathcal{B}$ under arbitrary unions, which we call the topology
\emph{generated by} $\mathcal{B}$.

We consider topological spaces over the set $\MTree$ of marked trees,
and the open sets are generated by the sets $\AB{p}$ for $p$ in some
fragment $\XP$.
\begin{theorem}
  If $\{\AB{p} \mid p \in \XP\}$ is the basis for a topology $\tau$ on
  $\MTree$, then a policy $\policy$ is fair with respect to $\XP$ if
  and only if $\AB{\policy}$ is open in $\tau$.
\end{theorem}
\begin{proof}
  If $\policy$ is fair, then $\AB{\policy} = \Approx_{\XP}(\AB{\policy})$.  Since each $\AB{p}$ is a (basic) open
  set, it is obvious that $\AB{\policy}$ is open.  Conversely, if $\AB{\policy}$
  is open, then $\AB{\policy} = \bigcup\{Y \in \tau \mid Y \subseteq
  \AB{\policy}\}$.  Thus, it suffices to show that $\bigcup\{Y \in \tau \mid
  Y \subseteq \AB{\policy}\} = \Approx_{\XP}(\AB{\policy})$.  The $\supseteq$ direction is immediate since every
  $\AB{p}$ is a basic open set.  For
  $\subseteq$, suppose $x \in \bigcup\{Y \in \tau \mid Y \subseteq
  \AB{\policy}\}$, that is, for some $Y\in \tau$ with $Y \subseteq \AB{\policy}$, we
  have $x \in Y$.  Any open set $Y$ is the union of basic open sets,
  so $x$ must be in some $\AB{p} \subseteq Y\subseteq \AB{\policy}$.  Hence
  $x \in\Approx_{\XP}(\AB{\policy})$. 
\end{proof}


\section{Main results}
\label{sec:results}

In this section we investigate fairness for different classes of
policies.  We first consider the simpler case of $\XPCDS$ policies and
show that they are always fair with respect to $\XPChild$.  Next, we consider fairness for
$\XPCDSF$ policies with respect to $\XPCF$ updates, and show that they
can be unfair only if they involve filters in negative rules.  We then
show that deciding fairness for such policies is \coNP-complete, and
conclude by discussing how our results extend to the general case of
$\XPFull$.

\subsection{Fairness for $\XPCDS$ policies} 

We call elements of $\XPChild$ \emph{linear paths}, and usually write
them as $\alpha,\beta$. For policies over $\XPCDS$, we consider the 
basis given by \emph{linear path sets} $\{\AB{\alpha} \mid \alpha \in \XPChild\}$.
\begin{proposition}
  The linear path sets partition $\MTree$ (and hence also form a basis for
  a topology on $\MTree$).
\end{proposition}
\begin{proof}
  Every point $(T,n) \in \MTree$ is in a linear path set: take $p$ to be
  the sequence of node labels along a path leading to $n$ in $T$.  Moreover, two linear path sets are
  either equal or disjoint.
\end{proof}
Consider the topology $\tau_1 = \tau_{\XPChild}$ generated by the
linear path sets.  Clearly, as for any partition topology, we have:
\begin{proposition}
  $\tau_1$ is closed under set complement.
\end{proposition}

Next, we show that any path in  $\XPCDS$ denotes an open set in
$\tau_1$, vi an auxiliary definition.
\begin{definition}
  We define the function  $\LPaths$ mapping  $p \in
  \XPCDS$ to a set of linear paths:
  \begin{eqnarray*}
       \LPaths(\kself::\phi) 
    &=& \{\kself::l \mid  l \in \semantics{\phi}\}\\
       \LPaths(\kchild::\phi) 
    &=& \{\kchild::l \mid  l \in \semantics{\phi}\}\\
    \LPaths(\kdescendant::\phi) &= &\LPaths(\kchild::*)^* \cdot \LPaths(\kchild::\phi)\\
    \LPaths(p/p') &=& \LPaths(p) \cdot \LPaths(p')
  \end{eqnarray*}
where $S \cdot T$ stands for $\{s/t \mid s \in S, t \in T\}$ and $S^*
= \bigcup_n S^n$.
\end{definition}
\begin{proposition}\label{prop:lpaths-fair}
  For every $p \in \XPCDS$, we have $P\AB{p} = \bigcup\{P\AB{\alpha}
  \mid \alpha \in \LPaths(p)\}$, and $\AB{p} = \bigcup\{\AB{\alpha}
  \mid \alpha \in \LPaths(p)\}$, hence $\AB{p}$ is open in $\tau_1$.
\end{proposition}
\begin{proof}
  The first part follows by induction on the structure of $p$.  The
  base cases for $\kchild::\phi$ and $\kself::\phi$ 
  are
  straightforward.  For a path $\kdescendant::\phi$, we reason as
  follows:
  \begin{eqnarray*}
   && P\AB{\kdescendant::\phi}\\
 &=& 
    \bigcup\{(T,n,m) \mid (n,m) \in E^+_T, \lambda_T(m) \in \semantics{\phi}\}\\
 &=& 
    \bigcup\{(T,n,m) \mid (n,m) \in E^+_T, \lambda_T(m) = l,l
    \in \semantics{\phi}\}\\
 &=& 
    \bigcup\{(T,n,m) \mid (n,n_1) \in E_T, \lambda_T(n_1) =
    \alpha_1,\ldots,\\
&&\qquad(n_k,m) \in E_T,\lambda_T(n_k) = \alpha_k,  \lambda_T(m) = l,l
    \in \semantics{\phi}\}\\
&=& 
\bigcup\{P\AB{\alpha/l} \mid \alpha \in \Sigma^*, l \in \semantics{\phi}\}\\
&=& \bigcup\{P\AB{\alpha_0} \mid \alpha_0 \in \LPaths(\kdescendant::\phi)\}
  \end{eqnarray*}
  For the fourth equation, observe that for any marked tree $(T,n)$
  there is a (possibly empty) path $\alpha$ formed of labels of nodes leading from the
  root of $T$ to $n$.  Conversely, for any $\alpha$
  there is a (linear) tree $T$ and node $n$ such that $\alpha$ is the
  list of labels of nodes from the root to $n$.

If the path is of the form $p/p'$, then we reason as follows:
\begin{eqnarray*}
&& P\AB{p/p'}\\
 &=& \{(T,n,m) \mid \exists k \in V_T. (T,n,k) \in P\AB{p},
(T,k,m) \in P\AB{p'}\}  \\
&=& \{(T,n,m) \mid \exists k \in V_T. (T,n,k) \in \bigcup\{P\AB{\alpha}
  \mid \alpha \in \LPaths(p)\},\\
&&\qquad
(T,k,m) \in \bigcup\{P\AB{\beta} \mid \beta \in \LPaths(p')\}\}  \\
 &=& \{(T,n,m) \mid \exists k \in V_T. (T,n,k) \in P\AB{\alpha}, \alpha
 \in \LPaths(p),\\
&&\qquad 
(T,k,m) \in P\AB{\beta}, \beta \in \LPaths(p')\}  \\
 &=& \{(T,n,m) \mid \exists k \in V_T. (T,n,k) \in P\AB{\alpha}, (T,k,m) \in P\AB{\beta},\\
&&\qquad 
\alpha
 \in \LPaths(p), \beta \in \LPaths(p')\}  \\
 &=& \bigcup \{P\AB{\alpha/\beta} \mid \alpha
 \in \LPaths(p), \beta \in \LPaths(p')\}  \\
 &=& \bigcup \{P\AB{\alpha_0} \mid \alpha_0
 \in \LPaths(p/p')\}  
\end{eqnarray*}

The second part is
  immediate since 
  \begin{eqnarray*}
    \AB{p} &=& \{(T,n) \mid (T,R_T,n) \in P\AB{p}\}\\
&=& \{(T,n) \mid (T,R_T,n) \in \bigcup\{P\AB{\alpha} \mid \alpha \in
\LPaths(p)\}\}\\
&=& \{(T,n) \mid (T,R_T,n) \in P\AB{\alpha}, \alpha \in
\LPaths(p)\}\\
&=& \bigcup \{\AB{\alpha} \mid \alpha \in \LPaths(p)\}
  \end{eqnarray*}
  which is a union of open sets in $\tau_1$.
\end{proof}

\begin{proposition}
  Every $\XPCDS$-policy $\policy$ denotes an open set in $\tau_1$.
\end{proposition}
\begin{proof}
  Clearly, the sets $\AB{\allowed},\AB{\denied}$ are open since they are unions of
  open sets.  Since $\tau_1$ is closed under complement, the set
  $\AB{\denied}$ is closed so $\AB{\policy}$ is
  open.
\end{proof}

\begin{corollary}
  Every $\XPCDS$-policy is fair with respect to $\XPChild$.
\end{corollary}

\subsection{Fairness for $\XPCDSF$ policies} 

Linear path sets are not rich enough to make all expressions in
$\XPCDSF$ denote open sets. For example, $\AB{a[b]}$ is not
open in $\tau_1$; if it were, then it would be expressible as a
(possibly infinite) union of basic open sets $\AB{\alpha}$.  However,
clearly the only $\alpha$ such that $\AB{\alpha}$ overlaps with
$\AB{a[b]}$ is $/a$, and $\AB{a[b]} \subsetneq \AB{a}$.  Thus,
updates based on linear paths are not sufficiently expressive for policies
involving filters.

Instead, we generalize to \emph{filter paths} $\XPCF$.  These paths
correspond in a natural way to marked trees $(T,n)$.  We adopt a
standard definition of  a 
\emph{tree homomorphism} $h : T \to U$ as a function mapping $V_T$
to $V_U$ such that 
\begin{enumerate}
\item 
$R_U = h(R_T)$,
 and 
\item for each $(v,w) \in E_T$ we have $(h(v),h(w)) \in E_U$, and
\item for each $v \in V_T$ we have $\lambda_T(v) = \lambda_U(h(v))$.  
\end{enumerate}
A marked tree
$(T,n)$ matches a tree $U$ at node $m$ (i.e., matches the marked tree
$(U,m)$) if there is a tree homomorphism $h :T \to U$ such that $h(n)
= m$.  We refer to such a homomorphism as a \emph{marked tree homomorphism}
$h : (T,n) \to (U,m)$, and write $\AB{T,n}$ for the set of all
homomorphic images of $(T,n)$.  If $p \in \XPCF$ corresponds to marked
tree $(T,n)$ then it is easy to show that $\AB{p} = \AB{T,n}$.

\begin{lemma}\label{lem:intersection}
  If $\AB{T,n}$ and $\AB{U,m}$ overlap, then there is a marked tree 
  $(V,k)$ such that $\AB{V,k} = \AB{T,n} \cap \AB{U,m}$.
\end{lemma}
This proof is technical, but straightforward; the details are in an appendix.
\begin{corollary}
    The sets $\{\AB{T,n} \mid (T,n) \in \MTree\}$ form a basis for a topology on $\MTree$.
\end{corollary}
Let $\tau_2$ be the topology generated by the
sets $\AB{T,n}$.
\begin{definition}
  The set $\FPaths(p)$ of filter paths of $p \in \XPCDSF$ is defined as
\begin{eqnarray*}
  \FPaths(ax::\phi) &=& \LPaths(ax::\phi)\\
\FPaths(p / p') &=& \FPaths(p) \cdot \FPaths(p')\\
\FPaths(p [q]) &=& \{p'[q'] \mid p' \in \FPaths(p),q'
\in \FPathsQ(q)\}\\
\FPathsQ(p) &=& \FPaths(p)\\
\FPathsQ(q_1 \kand q_2) &=& \{q_1' \kand q_2' \mid q_1' \in
\FPathsQ(q), q_2' \in  \FPathsQ(q')\}\\
\FPathsQ(\ktrue) &=& \{\ktrue\}
\end{eqnarray*}
\end{definition}
\begin{proposition}
For every $p \in \XPCDSF$, we have $P\AB{p} = \bigcup\{P\AB{p'} \mid
  p' \in \FPaths(p)\}$, and $Q\AB{q} = \bigcup\{Q\AB{q'} \mid q'
  \in \FPathsQ(q)\}$, hence $\AB{p}$ is open in $\tau_2$.
\end{proposition}
\begin{proof}
  We show by induction that for every $p \in \XPCDSF$, we have
  $P\AB{p} = \bigcup\{\AB{p'} \mid p' \in \FPaths(p)\}$. 
  The base cases are as in Prop.~\ref{prop:lpaths-fair}.
  The inductive step case for $p/p'$ is straightforward, following the
  same idea as in Prop.~\ref{prop:lpaths-fair}.  We give the inductive case for
  $p[q]$ as follows.
  \begin{eqnarray*}
  &&  P\AB{p[q]} \\
&=& \{(T,n,m) \mid (T,n,m) \in P\AB{p}, (T,m) \in    Q\AB{q}\}\\
&=& \{(T,n,m) \mid (T,n,m) \in \bigcup\{P\AB{p'} \mid p' \in
\FPaths(p)\},\\
&&\quad 
 (T,m) \in \bigcup\{Q\AB{q'} \mid q' \in \FPathsQ(q)\}\}\\
&=& \{(T,n,m) \mid (T,n,m) \in P\AB{p'} , p' \in
\FPaths(p), \\
&&\quad  (T,m) \in Q\AB{q'}, q' \in \FPathsQ(q)\}\\
&=& \{(T,n,m) \mid (T,n,m) \in P\AB{p'}, (T,m) \in Q\AB{q'}, \\
&&\quad  p' \in
\FPaths(p), q' \in \FPathsQ(q)\}\\
 &=& \bigcup\{\{(T,n,m) \mid (T,n,m) \in P\AB{p'}, (T,m) \in
 Q\AB{q'}\} \mid  \\
&&\quad  p' \in
\FPaths(p), q'\in \FPathsQ(q)\} \\
\end{eqnarray*}
\begin{eqnarray*}
&=& \bigcup\{P\AB{p'[q']} \mid  p' \in
\FPaths(p), q'\in \FPathsQ(q)\} \\
&=& \bigcup\{P\AB{p_0} \mid  p_0 \in
\FPaths(p[q])\} 
\end{eqnarray*}

For filters, the base case for $\ktrue$ is trivial.  
Suppose $q$ is a path existence test
$p$.  Then 
\begin{eqnarray*}
  Q\AB{p} &=& \{(T,n) \mid \exists m. (T,m,n) \in P\AB{p}\}\\
&=& \{(T,n) \mid \exists m. (T,m,n) \in \bigcup\{P\AB{p'} \mid p' \in \FPaths(p)\}\}\\
&=& \{(T,n) \mid \exists m. (T,m,n) \in P\AB{p'} , p' \in \FPaths(p)\}\\
&=& \bigcup\{\{(T,n) \mid \exists m. (T,m,n) \in P\AB{p'}\} \mid p' \in \FPaths(p)\}\\
&=& \bigcup\{P\AB{p'} \mid p' \in \FPaths(p)\}\\
&=& \bigcup\{Q\AB{q'} \mid q' \in \FPathsQ(q)\}
\end{eqnarray*}

Finally, if $q$ is a conjunction $q_1 \kand q_2$, then we reason as
follows:
\begin{eqnarray*}
&&  Q\AB{q_1 \kand q_2}\\
 &=& Q\AB{q_1} \cap Q\AB{q_2}\\
 &=& \bigcup\{Q\AB{q_1'} \mid q_1' \in \FPathsQ(q_1)\} \cap
 \bigcup\{Q\AB{q_2'} \mid q_2' \in \FPathsQ(q_2)\}\\
 &=& \bigcup\{Q\AB{q_1'} \cap Q\AB{q_2'} \mid q_1' \in
 \FPathsQ(q_1),q_2' \in \FPathsQ(q_2)\}\\
 &=& \bigcup\{Q\AB{q_1'\kand q_2'} \mid q_1' \in \FPathsQ(q_1),q_2' \in \FPathsQ(q_2)\}\\
 &=& \bigcup\{Q\AB{q'} \mid q' \in \FPathsQ(q_1 \kand q_2)\}
\end{eqnarray*}

The argument that $\AB{p}$ is open is similar to that for Prop.~\ref{prop:lpaths-fair}.
\end{proof}
Note that $\tau_2$ is \emph{not} closed under complement; for example,
the complement of $\AB{a[b]}$ is not open.  This is, intuitively, why
unfair policies (such as Example~\ref{ex:unfair}) exist for
$\XPCDSF$.  However, we do have:
\begin{theorem}
  If $\policy = (\ds,\crp,\allowed,\denied)$ where $\allowed \subseteq \XPCDSF$
  and $\denied \subseteq \XPCDS$ then $\policy$ is fair with respect
  to $\XPCF$.
\end{theorem}
\begin{proof}
  Any policy whose negative rules denote a \emph{closed} set is fair,
  since open sets are preserved by removing closed sets.  All sets
  built from paths in $\XPCDS$ are closed in $\tau_1$, hence also
  closed in $\tau_2$, so policies over $\XPCDSF$ with no filters
  in negative rules are fair.
\end{proof}
The converse does not hold; for example, the
policy %
\[(-,-,\{\udelete{/a}\},\{\udelete{/a[b]},\udelete{//*}\})\]
is fair even though it involves negative filter paths, because the
negative rule $\udelete{//*}$ subsumes the
negative rule $\udelete{/a[b]}$.

\subsection{Complexity of fairness for \XPCDSF}
In this section we show that fairness for $\XPCDSF$ policies with respect to
$\XPCF$ is $\coNP$-complete.  These results draw upon  some material and notation from Miklau and Suciu's study of XPath query containment~\cite{miklau04jacm}, which we first review.

A \emph{tree pattern} is a structure $P= (V_P,C_P, D_P, R_P,
\lambda_P)$ such that $C_P \subseteq D_P \subseteq V_P \times V_P$ and $R_P \in
V_P$ and $\lambda_P : V_P \to \Sigma \cup \{*\}$.  In addition, $(V_P,
D_P, R_P, \lambda_P)$ forms an ordinary tree; the edges in
$C_P$ are called child edges and those in $D_P$ are called descendant
edges.  We can think of such a pattern $P$ as a tree with edges
labeled by axes $\kchild$ or $\kdescendant$, nodes labeled by node
tests $l$ or $*$, and with a distinguished node $n$.  A marked tree
pattern $(P,n)$ is a tree pattern with a specific marked node $n \in V_P$.
Path expressions $p
\in \XPCDSF$ are equivalent to tree patterns $(P,n)$, and 
an ordinary marked tree $(U,m)$ is essentially the same as a tree pattern
that has no $\kdescendant$ edges or $*$ nodes.  

A tree $T$ \emph{matches} a tree pattern $P$ if there is a function $h : V_P
\to V_T$ such that 
\begin{enumerate}
\item $h(R_P) = R_T$
\item for all $(v,w) \in C_P$ we have $(h(v),h(w)) \in E_T$
\item for all $(v,w) \in D_P$ we have $(h(v),h(w)) \in E^+_T$
\item for all $v \in V_P$ we have $\lambda_T(h(v)) \in \semantics{\lambda_P(v)}$.
\end{enumerate}
We then say that $h : P \to T$ is a \emph{tree
  pattern embedding}.  Similarly, a marked tree pattern $(P,k)$ matches a
marked tree $(T,n)$ provided there is a tree pattern embedding $h : P
\to T$ such that $h(k) = n$; we then write  $h :
(P,k) \to (T,n)$.

The \emph{star
  length} of a pattern or path is the length of the longest sequence
of $\kchild::*$ steps.  A $(u_1,\ldots,u_d)$-extension of a path $p$
with $d$ descendant edges $(e_1,\ldots,e_d)$ is a path $p[\bar{u}]$
where each descendant edge $e_d$ has been replaced with a path of $u_i$
$\kchild::*$ steps.  A canonical instance of $p$ is obtained by
substituting all occurrences of $*$ in a
$(u_1,\ldots,u_d)$-extension of $p$ with some fresh symbol $z$; this
is written $s^z(p[\bar{u}])$.  The set of canonical instances of $p$
(with replacement symbol $z$)
is $mod^z(p)$, and the set of such instances where the extension
lengths $u_i$ are bounded by $k$ is $mod^z_k(p)$.  
Any containment problem $p \sqleq p'$ is satisfied
if and only if there is no counterexample in $mod^z_{w'+1}(p)$ where
$w'$ is the star length of $p'$.  This implies that the size of a
counterexample is bounded by a polynomial determined by $p$ and $p'$,
since $w' \leq |p'|$ and $d \leq |p|$.

Miklau and Suciu~\cite{miklau04jacm} also give polynomial algorithms
for testing containment of $p,p'$ in special cases. Specifically when
$p \in\XPCSF$, we can test whether $p\sqleq p'$ in polynomial time.
Furthermore, Miklau and Suciu~\cite{miklau04jacm} give an
algorithm for containment that is polynomial when the number of
descendant steps in $p$ is at most $d$.  

\paragraph*{\coNP-Hardness}
Hardness follows by reduction from path containment:

\begin{theorem}\label{thm:conp-hard}
  Determining fairness for $\XPCDSF$ policies with respect to $\XPCF$ is \coNP-hard.
\end{theorem}
\begin{proof}
 For hardness, the
idea is to encode a containment problem $p \sqleq p'$ as a fairness
problem $(-,-,\{/{*}[p]\},\{/{*}[p']\})$.  Let $p,p'$ be given and define
policy $\policy = (-,-, \{/{*}[p]\},\{/{*}[p']\})$.  Suppose $p \sqleq
p'$.  Then $/{*}[p] \sqleq /{*}[p']$ so $\AB{\policy} = \emptyset$, which
is obviously fair.  Conversely, suppose $p \not\sqleq p'$.  Then there
must be some marked tree $(T,n) \in \AB{p} - \AB{p'}$.  Moreover, all
paths in $\XPCDSF$ are satisfiable so choose some $(T',n') \in
\AB{p'}$.  Let $v$ be a fresh vertex identifier not appearing in $T$ or $T'$,
and assume without loss of generality that $T$ and $T'$ do not
overlap.  Form trees $(U,n)$ and $(U',n)$ as follows: 
\begin{eqnarray*}
  U &=& (V_T \cup \{v\}, E_T \cup \{(v,R_T)\}, v, \lambda_T \cup
\{(v,a)\})\\
 U' &=& (V_U \cup V_{T'}, E_U \cup E_{T'} \cup \{(v,R_{T'})\}, v,
 \lambda_U \cup \lambda_{T'})
\end{eqnarray*}
In other words, $U$ is a copy of $T$ placed under a new root node
labeled $a$, and $U'$ is $U$ extended with a copy of $T'$ under the
root.
Clearly, $U$ matches $/{*}[p]$ and not $/{*}[p']$, whereas $U'$ matches both $/{*}[p]$ and
$/{*}[p']$.  Moreover, it is also easy to see that $U'$ is a homomorphic
image of $U$ (by an inclusion homomorphism).  Thus, $U$ and $U'$ are
witnesses to the fact that $\AB{\policy}$ is not closed under homomorphic
images, which implies $\AB{\policy}$ is not an open set so $\policy$ is unfair.

Thus, $p \sqleq p'$ holds
if and only if $\policy$ is fair.  Since containment of $\XPCDSF$
paths is $\coNP$-hard, fairness is also $\coNP$-hard.
\end{proof}

\paragraph*{\coNP-Completeness}
For \coNP-completeness we need the following lemma:
\begin{proposition}\label{prop:closed-hom}
  A set $Y \subseteq \MTree$ is open in $\tau_2$ if and only if $Y$ is
  closed under homomorphic images; that is, for all $(T,n) \in Y$ and
  $h :  (T,n) \to (U,m)$ we have $(U,m) \in Y$.
\end{proposition}

\begin{proof}
  If $Y$ is open, then suppose $(T,n)$ is a point in $Y$ and $h :
  (T,n) \to (U,m)$.  Since $Y$ is the union of basic open sets, there
  must be some $(V,k)$ such that $(T,n) \in \AB{V,k} \subseteq Y$.
  That is, there is a homomorphism $g : (V,k)\to (T,n)$.  Hence, $h
  \comp g : (V,k)\to (U,m)$ so $(U,m) \in \AB{V,k} \subseteq Y$, as
  desired. 

  Conversely, if $Y$ is closed under homomorphic images, then we will
  show that $Y = \bigcup\{\AB{T,n} \mid (T,n) \in Y\}$.  The
  $\subseteq$ direction is immediate since $(T,n) \in \AB{T,n}$; on
  the other hand, for each $(T,n) \in Y$, it follows that $\AB{T,n}
  \subseteq Y$ since each element of $\AB{T,n}$ is a homomorphic
  image of $(T,n) \in Y$. Hence, $Y$ is a union of basic open sets, so
  it is open. 
\end{proof}

The basic idea of the proof of the \coNP\ upper bound is as follows.
We need to show that for any policy $\policy$, it suffices to consider
a finite set of trees (of size bounded by a polynomial in the policy
size) in order to decide whether $\policy$ is closed under
homomorphisms.  To illustrate, let a counterexample consisting of
trees $(T,n)$ and $(T',n)$ and homomorphism $h : (T,n) \to (T',n)$ be
given, such that $(T,n) \in \AB{\policy}$ and $(T',n) \not\in
\AB{\policy}$.

  First, consider a deny--deny policy, so that $(T,n) \in
  \AB{\allowed} - \AB{\denied}$ and $(T',n) \notin \AB{\allowed} - \AB{\denied}$.   Since $\AB{\allowed}$ is open and $(T,n) \in
  \AB{\allowed}$, it follows that $(T',n) \in \AB{\allowed}$, so we must have
  $(T',n) \in \AB{\denied}$.  Moreover, there must exist paths $p\in \allowed$ and
  $p'\in \denied$ such that $(T,n) \in \AB{p}$ and $(T',n) \in \AB{p'}$.  It
  is easy to see that $(T',n) \in \AB{p}$ also, while $(T,n) \not\in
  \AB{\denied}$ means that $(T,n)$ does not match any path in
  $\AB{\denied}$.  Observe that $(T,n)$ and $(T',n)$ could be much
  larger than $\policy$.  It suffices to show that we can shrink
  $(T,n)$ and $(T',n)$ to a small counterexample by deleting nodes
  and edges that do not affect satisfiability of $p,p'$, using similar
  techniques to those used by Miklau and Suciu~\cite{miklau04jacm}.
  They considered how to shrink a counterexample to the containment
  problem $p \sqleq p'$, consisting of a single tree, whereas we need
  to shrink $(T,n)$, $(T',n)$ and $h$ while ensuring that $h$ is
  still a homomorphism, and also that the
  shrinking process does not cause the first tree to satisfy some
  other path in $\denied$.  Thus, it suffices to search for small
  ($O(|\policy|^3)$) counterexamples.

  The reasoning for other kinds of policies (allow--allow, etc.) is
  similar. This in turn gives a \coNP-time decision procedure to
  determine fairness: first we guess a pair of trees $(T,n)$, $(T',n)$
  with $h : (T,n) \to (T',n)$ and $|T|,|T'| \leq O(|\policy|^3)$, then check
  whether $(T,n) \in \AB{\policy}$ and $(T',n) \notin
  \AB{\policy}$. If no such counterexamples exist, then $\policy$ is
  fair.

The proof makes use of the following facts which are immediate or
proved by Miklau and Suciu~\cite{miklau04jacm}.
\begin{lemma}[\cite{miklau04jacm}]\label{lem:shrinking}
  \begin{enumerate}
  \item   If $h : (P,n) \to (T,n)$ is an embedding witnessing that
    $(T,n)$ matches some path $p$ with $\AB{p} = \AB{P,n}$, and
    $(T',n)$ is a subtree of $T$ such that $rng(h) \subseteq V_{T'}$,
    then $h : (P,n) \to (T',n)$ witnesses that $(T',n) $ matches $p$.
  \item If $(T,n) \notin \AB{p}$ and $(T',n)$ is obtained by
    removing any subtree from $(T,n)$ then $(T',n) \notin \AB{p}$.
  \item If $(T,n)$ contains a path of child steps of
    length $> w+1$, where $w$ is the star length of $p$, $z$ is not present in $p$, and each node
    along the path is labeled $z$, then we can form $(T',n)$ by
    removing one of the steps, such that  $(T,n) \in \AB{p} \iff (T',n) \in \AB{p}$.
  \end{enumerate}
\end{lemma}

\begin{theorem}\label{thm:conp-decidable}
  Deciding whether a policy in $\XPCDSF$ is fair with
  respect to $\XPCF$ is \coNP-complete.
\end{theorem}
\begin{proof}
  For deny--deny policies, suppose $(T,n) \in \AB{\policy}$ and
  $(T',n) \not\in \AB{\policy}$ where $h : (T,n) \to (T',n)$.   This
  implies that there exists $p \in \allowed, p' \in \denied$ with $(T,n) \in
  \AB{p} - \AB{\denied}$ and $(T',n) \in \AB{p'}$.  Construct expression
  $p''$ such that $(T',n) \in \AB{p''} \subseteq \AB{p} \cap \AB{p'} $ and $|p''| \leq |p|+|p'|$.   Without loss of
  generality assume wherever not required by matching $p,p'$, the labels of $T,T'$ are some $z \in \Sigma$
  not appearing in $\policy$.  (Relabeling $T,T'$ in this way cannot
  affect whether they satisfy $\policy$ since $z$ does not appear
  there).  Then, using Lem.~\ref{lem:shrinking} we can shrink $T,T'$ and $h$ by removing
  subtrees that are not needed to ensure that $T,T'$ match $p, p''$
  respectively; moreover, we can maintain $h$ so that it remains a
  homomorphism through this process.  This yields trees where every
  leaf node is needed for matching $p,p''$, but where there may still
  exist long chains of $z$s that are only needed to match descendant
  steps in $p$ or $p''$.  However, again using
  Lem.~\ref{lem:shrinking} we can remove $z$-labeled
  nodes from any chains longer than $W+1$, where $W$ is the maximum
  star length of any path in $\policy$, and we can maintain $h$ so
  that it remains a homomorphism.  Call the resulting trees $(U,n),
  (U',n)$.  By the above lemma, $(U,n) \in \AB{p} - \AB{\denied}$ and
  $(U',n) \in \AB{p''}$ still hold since $W$ is larger than the star
  height of any path in $\denied$.  Moreover, $U,U'$ have at most
  $(|p|+|p'|)(W+1)$ nodes because any two nodes needed for matching
  $p,p''$ can be separated by a chain of at most $W+1$ $z$-nodes.

  For allow--allow policies, the reasoning is slightly different.  If
  $(T,n) \in \AB{\policy}$ but $(T',n) \notin\AB{\policy}$ then
  $(T,n)$ cannot match $\AB{\allowed}$ because if it did, then so would
  $(T',n)$.  Thus, $(T,n)$ does not match $\denied$ either.  Similarly, 
  $(T',n)$ must match some negative rule $p\in \denied$ and no positive
  rules $\in \allowed$.
  The rest of the argument is similar; we obtain a small
  counterexample by replacing unimportant node labels with some fresh
  $z$, removing subtrees, and shortening long chains of $z$s.

  The allow--deny and deny--allow cases are special cases of the
  above.  Hence, in any case, to decide whether $\policy$ is
  homomorphism-closed it suffices to check for counterexamples among
  trees of size bounded by   $(|p|+|p'|)(W+1)$.
  \end{proof}

\subsection{Polynomial-time static enforcement}

Fairness ensures static enforceability, but the problem of checking
whether an update operation is statically allowed by a policy can
still be expensive.  Consider the common case of a deny--deny policy.
An update $U$ is statically allowed if and only if it is contained in
$\allowed$, and does not overlap with $\denied$.  Overlap testing is decidable in
polynomial time~\cite{hammerschmidt05ideas}, but containment of XPath
expressions involving unions is \coNP-complete.  This high complexity
is however dependent only on the policy and update size, not the size
of the data, so may still be acceptable in practice; also,
efficient-in-practice solvers are being developed for XPath
containment and overlap tests~\cite{geneves07pldi}.

We can take advantage of several observations to obtain efficient
algorithms for special cases.  First, we identify classes of XPath
queries satisfying the following \emph{union decomposition} property:
\begin{eqnarray}
  \label{eq:union}
  p \sqleq p_1 | \cdots | p_n & \iff & p \sqleq p_1 \vee \cdots \vee p
  \sqleq p_n
\end{eqnarray}

We need some auxiliary lemmas:


\begin{lemma}\label{lem:decomp-filter}
  Suppose $p \in \XPCF$ and $Y_1,\ldots,Y_n$ are open sets in
  $\tau_2$.  Then $\AB{p} \subseteq Y_1 \cup \cdots \cup Y_n$ if and
  only if $\AB{p} \subseteq Y_1 \vee \cdots \vee \AB{p} \subseteq Y_n$.
\end{lemma}
\begin{proof}
  Clearly $\AB{p}$ is nonempty, and as discussed in
  Sec.~\ref{sec:results} $\AB{p}$ contains a tree $(T,n)$ such that
  $\AB{p} = \AB{T,n}$.  Thus, $(T,n) \in Y_1 \cup \cdots \cup
  Y_n$, so for some $i$ we have $(T,n) \in Y_i$.  By
  Prop.~\ref{prop:closed-hom} we know that $Y_i$ is closed under
  homomorphic images of $(T,n)$, but the set of homomorphic images of
  $(T,n)$ is precisely $\AB{T,n} = \AB{p}$. 
\end{proof}
\begin{corollary}
    Suppose $p \in \XPCF$ and $p_1,\ldots,p_n \in \XPCDSF$.  Then $p \sqleq p_1 | \cdots | p_n$ if and
  only if $p \sqleq p_1 \vee \cdots \vee p \sqleq p_n$.
\end{corollary}

Next, in order to prove union decomposition for containment problems
whose left-hand side involves wildcards, we introduce relabeling
functions $\rho : \Sigma \to \Sigma$.  We define $\rho(T)$ in the
obvious way: specifically,
\[\rho(T) = (V_T, E_T, R_T, \rho \circ \lambda_T)\;.\]
Similarly, $\rho(T,n) = (\rho(T),n)$ and $\rho(T,n,m) =
(\rho(T),n,m)$; furthermore if $Y$ is a set of (marked) trees then
$\rho(Y) = \{\rho(y) \mid y \in Y\}$.
\begin{definition}
  Suppose $C \subseteq \Sigma$ is finite.  We say that $\rho$
  \emph{fixes} $C$ if $\rho(c) = c$ for each $c \in C$.  A set $Y$ of
  trees, marked trees or doubly marked trees is called
  \emph{$C$-invariant} if for all $\rho$ fixing $C$, we have $\rho(Y)
  \subseteq Y$.  In other words, $Y$ is closed under relabelings that
  replace labels in $\Sigma - C$ with arbitrary labels.
\end{definition}

We define the function  $\labels$  mapping each path to the finite set of labels appearing in it, and likewise $\labelsQ$ mapping each filter to its finite set of labels:
\begin{eqnarray*}
  \labels(ax::a) &=& \{a\}\\
\labels(ax::*) &=& \emptyset\\
\labels(p/p') &=& \labels(p) \cup \labels(p')\\
\labels(p[q]) &=& \labels(p) \cup \labelsQ(q)\\
\labelsQ(\ktrue) &=& \emptyset\\
\labelsQ(q\kand q') &=& \labelsQ(q) \cup \labelsQ(q')\\
\labelsQ(p) &=& \labels(p)
\end{eqnarray*}
Thus, for example, $\labels(/a//{*}/b[c]) = \{a,b,c\}$ is the set of
specific labels appearing in $/a//{*}/b[c]$.  The semantics of a path
$p$ is $\labels(p)$-invariant:

\begin{lemma}\label{lem:invariant}
  For any $p \in \XPCDSF$:
  \begin{enumerate}
  \item $P\AB{p}$ is $\labels(p)$-invariant.
  \item $Q\AB{q}$ is $\labelsQ(q)$-invariant.
\item $\AB{p}$ is $\labels(p)$-invariant.
  \end{enumerate}
\end{lemma}
\begin{proof}
  Given $p$, define $C$ to be the set of all node labels from $\Sigma$
  appearing in $p$. We proceed to prove parts (1,2) by simultaneous
  induction on the structure of path expressions and filters.

  First, consider the case $ax::a$.  Observe that $\labels(ax::a) =
  \{a\}$, so let $\rho$ be given such that $\rho(a) = a$ and suppose
  $(T,n,m) \in P\AB{ax::a}$.  Then $(n,m) \in A\semantics{ax}(T)$ and
  $\lambda_T(m) = a$.  Since the semantics of axis steps depends only
  on $E_T$ we know that $(n,m) \in A\semantics{ax}(\rho(T))$, and
  since $\rho(a) = a$ we know that $\lambda_{\rho(T)}(m) =
  \rho(\lambda_T(m)) = \rho(a) = a$, which implies $(\rho(T),n,m) \in
  P\AB{ax::a}$.

  Next, for the case $ax::*$, observe that $\labels(ax::*) =
  \emptyset$, so we must consider arbitrary renamings $\rho$.  Let
  $\rho$ be given and suppose $(T,n,m) \in P\AB{ax::*}$.  Then $(n,m)
  \in A\semantics{ax}(T)$, and since the semantics of axis steps
  depends only on $E_T$ we can conclude $(n,m) \in
  A\semantics{ax}(\rho(T))$, which implies $(\rho(T),n,m) \in
  P\AB{ax::*}$.

  The cases for $p/p'$, $p[q]$, and filters are straightforward.  For
  example, let $\rho$ fixing $\labels(p/p') = \labels(p)
  \cup \labels(p')$ be given, and suppose $(T,n,m) \in P\AB{p/p'}$.  Then there is some $k$
  such that $(T,n,k) \in P\AB{p}$ and $(T,k,m) \in P\AB{p'}$.
  Clearly, $\rho$ fixes $\labels(p)$ and $\labels(p')$ so
  $(\rho(T),n,k) \in P\AB{p}$ by induction and similarly $(\rho(T),k,m)
  \in P\AB{p'}$.  So, we can conclude that $(\rho(T),n,m) \in P\AB{p/p'}$.

Finally, for part (3) if $(T,n) \in \AB{p}$ then $(T,R_T,n) \in P\AB{p}$ so
$(\rho(T),R_T,n) \in P\AB{p}$ and we can conclude that $(\rho(T),n)
\in \AB{p}$.
\end{proof}

\begin{lemma}\label{lem:decomp-invariant}
  Suppose $p \in \XPCSF$ and $Y_1,\ldots,Y_n$ are open in
  $\tau_2$ and assume that each $Y_i$ is $C$-invariant for some fixed
  $C$.  Then $\AB{p} \subseteq Y_1 \cup \cdots \cup Y_n$ if and only if
  $\AB{p} \subseteq Y_1 \vee \cdots \vee \AB{p} \subseteq Y_n$.
\end{lemma}
\begin{proof}
  Recall that we assume $\Sigma$ is infinite, so choose an infinite
  sequence $x_1,x_2,\ldots$ of elements of $\Sigma - C$.  Form a new
  path expression $p'$ from $p$ by replacing each $*$ occurring in $p$
  with a distinct $x_i$.  For example, if $p = /a/{*}[b/{*}]$ then $p' =
  /a/x_1[b/x_2]$.  Clearly, by construction $\AB{p'} \subseteq \AB{p}
  \subseteq Y_1 \cup \cdots \cup Y_n$.  Therefore, by
  Lem.~\ref{lem:decomp-filter} there must be some $i$ such that
  $\AB{p'} \subseteq Y_i$; fix such an $i$.  

  We now show that for any $(T,n) \in \AB{p}$ there exists a tree $U$
  and relabeling $\rho$ such that $(U,n) \in \AB{p'}$ and $\rho(U) =
  T$ and $\rho$ fixes $C$.  Let $(T,n) \in \rho$ be given, and let
  $\{m_1,\ldots,m_k\} = V_T$ be some enumeration of the $k$ vertices
  of $T$.  Let $(P,v)$ be a tree pattern corresponding to $p$, and let
  $h : (P,v) \to (T,n)$ be an embedding witnessing the fact that
  $(T,n)$ matches $p$.  Define $U$ and $\rho$ as follows:
\begin{eqnarray*}
 U &=& (V_T, E_T, R_T, \lambda')\\
  \lambda'(m_i)& = &\left\{\begin{array}{ll}
\lambda_T(m_i) & \exists m'.h(m') = m_i, \lambda_{P}(m') = a\\
x_i & \text{otherwise}
\end{array}\right.
\\
\rho(a) &=& \left\{\begin{array}{ll}
\lambda_T(m_i) & a = x_i\\
a & \text{otherwise}
\end{array}\right.
\end{eqnarray*}
That is, $U$ has the same nodes and edges as
  $T$, and its node labels are equal to
  those of $T$ for nodes that match a fixed label $a$ in $p$ (i.e.,
  when $h(m_i) = m'$ and $\lambda_{P}(m_i) = a$), and the
  labels of nodes $m_i$ matching occurrences of $*$ are reassigned to the
  corresponding $x_i$ in $p'$. Also, $\rho$ maps each $x_i$ to the
  corresponding label $a$ in $T$, so that by construction $\rho(U) =
  T$.    Since the
  $x_i$ are chosen from outside $C$ it follows that $\rho$ fixes $C$
  by construction.  

Now we show that $\AB{p} \subseteq Y_i$.  Let $(T,n)$ be an element of
$\AB{p}$ and let $U,\rho$ be constructed as above, so that $\rho(U) =
T$ and $(U,n) \in \AB{p'}$ and $\rho$ fixes $C$.  Clearly, $(U,n) \in
\AB{p'} \subseteq Y_i$.  Therefore, by the assumption that each $Y_i$ is
  $C$-invariant, we have that $(T,n) = (\rho(U),n) \in \rho\AB{p'}
  \subseteq \rho(Y_i)
  \subseteq Y_i$. 
\end{proof}

\begin{lemma}\label{lem:decomp}
  The containment problem $p \sqleq p_1 | \cdots | p_n$ satisfies
  union decomposition provided that $p \in \XPCSF$ and $p_i \in \XPCDSF$.
\end{lemma}
\begin{proof}
  From Lem.~\ref{lem:invariant} we have that all of the sets
  $\AB{p_i}$ are $\labels(p_i)$-invariant, so they are all
  $\bigcup_i \labels(p_i)$-invariant.  Thus, by
  Lem.~\ref{lem:decomp-invariant} we must have $\AB{p} \subseteq
  \AB{p_1} \cup \cdots \cup \AB{p_n}$ if and only if $\AB{p} \subseteq
  \AB{p_i}$ for some $i$.  This is equivalent to union decomposition
  for the problem $ p \sqleq p_1 | \cdots | p_n$. 
\end{proof}
Note that for this proof, the assumption that $\Sigma$ is
infinite was necessary: otherwise, if $\Sigma = \{a_1,\ldots,a_n\}$, then $/{*}
\sqleq /a_1 | \cdots | /a_n$ holds but does not satisfy union decomposition.

\begin{theorem}\label{thm:enforcement-ptime}
  Static enforcement of update capabilities in $\XPCSF$ is
  checkable in $\PTIME$ for any fixed policy $\policy$ over $\XPCDSF$.
\end{theorem}
\begin{proof}
  We consider the deny--deny case where $\policy = (-,-,\allowed,\denied)$.
  Consider an update capability $U$ characterized by a path $p \in
  \XPCSF$.  We need to ensure that $\AB{p} \subseteq \AB{\allowed}$
  and $\AB{p} \cap \AB{\denied} = \emptyset$.  By Lemma~\ref{lem:decomp}
  the first part can be checked by testing whether $p \sqleq p_i$ for
  each $p_i \in \allowed$.  Each such test can be done in polynomial
  time since $p \in \XPCSF$. The second part amounts to checking
  that $p$ does not overlap with any element of $\denied$, which also takes
  polynomial time.

The allow--allow case is similar, but more involved.  Given $p$, we
first check whether it overlaps with any elements of $\denied$.  For each
$p' \in \denied$ such that  $p$
overlaps with $p'$, we need to check whether $\AB{p} \cap \AB{p'}
\subseteq \AB{\allowed}$.  The intersection of two paths in $\XPCSF$ can
be expressed by another path in $\XPCSF$ and this can be computed
in $\PTIME$; the required containment checks are also in $\PTIME$ as per
Lemma~\ref{lem:decomp}.  Hence, allow--allow policies can also be statically
enforced in \PTIME.  Other policies are special cases.
\end{proof}

Unfortunately, union decomposition does not hold for problems where $p
\in \XPCDS$.  For example, $/a//b \sqleq (/a/b) ~|~ (/a/{*}//b)$
holds, but neither $/a//b \sqleq /a/b$ nor $/a//b \sqleq /a/{*}//b$
holds.  In any case, even without union, containment of $\XPCDSF$
paths is \coNP-complete. As noted earlier, Miklau and Suciu showed
that containment is decidable in polynomial time if the number of
descendant steps in $p$ is bounded. Using an adaptation of this
result, together with Theorem~\ref{thm:enforcement-ptime}, we can extend
this result to handle problems of the form $p \sqleq p_1 | \cdots |
p_n$ where all paths are in $\XPCDSF$ and $p$ has at most $d$
descendant steps:

\begin{corollary}\label{cor:ptime}
    Static enforcement of update capabilities in $\XPCDSF$
    having at most $d$ descendant steps is
  checkable in $\PTIME$ for any policy $\policy$ over $\XPCDSF$.
\end{corollary}

\begin{proof}
  Given a problem $p \sqleq p_1|\ldots|p_n$, consider all expansions
  $p[\bar{u}]$ where $u_i \leq W+1$, where $W$ is the maximum star
  length of the paths $p_1,\ldots,p_n$.  There are at most $m=(W+2)^d$
  such expansions where $d$ is the number of descendant steps in $p$.
  Each of these paths is in $\XPCSF$ so by Lem.~\ref{lem:decomp} we can check in \PTIME\  whether
  all are contained in $p_1|\cdots|p_n$.  If not, then clearly $p$
  itself is not contained in $p_1|\cdots|p_n$.  Conversely, if $p$ is
  not contained in $p_1|\cdots|p_n$, then (using
  Lem.~\ref{lem:shrinking}) we can find a small counterexample that
  matches one of the $p[\bar{u}]$, which implies that the algorithm
  will detect non-containment for this $p[\bar{u}]$.
\end{proof}



\section{Discussion}
\label{sec:discussion}

\subsection{Generalizations}

In the previous section we have simplified matters by considering only
deletion capabilities; we also have not discussed attribute equality
tests. Our framework extends to policies over the full language
$\XPFull$ and to policies consisting of multiple different kinds of
operations (insert, delete, rename, replace). To handle multiple kinds
of operations, we need to consider the topologies over the set of
pairs $(T,u)$ of trees and atomic operations $u \in
\AtomicUpdates(T)$. Attribute steps and value tests complicate
matters: because attribute values must be unique, the policy
$(-,-,\{/a[@b=c]\},\{/a[@b=d]\})$ is fair. Verifying this requires
taking the uniqueness constraint into account, or more generally,
testing containment or overlap modulo key constraints. This can be
done using more expressive logics for XPath over data
trees~\cite{david12tods,david13icdt} or general-purpose
solvers~\cite{geneves07pldi}. However, the \coNP-hardness proof given
earlier is not applicable if only attribute-based filters are allowed,
so it may be possible to check fairness in the presence of negative
attribute tests in \PTIME.

Fairness can also be affected by the presence of a DTD or schema that
constrains the possible trees.  It is easy to see that a policy that
is fair in the absence of a schema remains fair if we consider only
valid documents.  On the other hand, an unfair policy may become
fair in the presence of a schema or other constraints (as illustrated
above using attributes).  For example, the unfair policy from
Ex.~\ref{ex:unfair} becomes fair if the schema eliminates uncertainty
as to whether $a$ has a $b$ child.  This can happen either if $a$
cannot have any $b$ children or always has at least one.  However, checking containment and satisfiability often become more difficult when a DTD is present~\cite{neven06lmcs,benedikt08jacm}.

The fairness picture changes if we consider extensions to the XPath
operations allowed in the update requests.  For example, for
$\XPFull$-policies, it appears possible to recover fairness by
allowing negation in filter expressions (e.g.~$a[not(b)]$).  However,
XPath static analysis problems involving negation are typically not in
$\PTIME$~\cite{neven06lmcs,benedikt08jacm,tencate09jacm}.  Thus, there
is a tradeoff between the complexity of determining that a policy is
fair and the complexity of statically enforcing fairness, governed by
the expressiveness of the set $XP$ of XPath expressions allowed in
update capabilities.  The more expressive $XP$ is, the easier it is to
check fairness and the harder it is to enforce the policy.

\subsection{Implications}

Having established some technical results concerning policy fairness
and the complexity of determining fairness and of static enforcement
with respect to fair policies, what are the implications of these
results? We believe that there are three main messages:
\begin{itemize} 
\item Policies without filters are always fair.
  However, such policies may not be sufficiently expressive for
  realistic situations; for example, the policy in
  Figure~\ref{fig:policy} would become much too coarse if we removed
  the filters.
  Policies with filters only in positive rules are also always fair,
  and are more expressive; for example, the policy in
  Figure~\ref{fig:policy} is in this fragment.  Therefore, policy
  authors can easily ensure fairness by staying within this fragment.

\item Checking policy
  fairness for policies with
  filters in negative rules may be computationally intensive; it may
  be worthwhile investigating additional heuristics or static analyses
  that can detect fairness for common cases more efficiently. Also, as
  discussed in the previous section, it may be possible to check
  fairness for policies with negative attribute tests in \PTIME.

\item We
  established that for relatively tame update capabilities (with
  limited numbers of descendant axis steps), static enforcement
  remains in PTIME. Static enforcement depends directly on the
  complexity of checking containment and overlap problems. Containment
  checking $p \sqleq p'$ is not symmetric in $p$ and $p'$, so it may
  be profitable to investigate ways to make policies richer while
  retaining fairness with respect to less expressive classes of
  updates. 
\end{itemize}


\section{Related Work}
\label{sec:related}

Most prior work on enforcing fine-grained XML access control policies
has focused on dynamic enforcement strategies. As discussed in the
introduction, previous work on filtering, secure query evaluation and
security views has not addressed the problems that arise in update
access control, where it is important to decide whether an operation
is allowed before performing expensive  updates.

Murata et al.~\cite{murata06tissec} previously considered static
analysis techniques for rule-based policies, using regular expressions
to test inclusion in positive rules or possible overlap with negative
rules, but their approach provides no guarantee that static
enforcement is fair; their static analysis was used only as an
optimization to avoid dynamic checks.  Similarly, Koromilas et
al.~\cite{koromilas09sdm} employed static analysis techniques to speed
annotation maintenance in the presence of updates.  In contrast, our
approach entirely obviates dynamic checks.

The \emph{consistency} problem for XML update access control policies
involves determining that the policy cannot be circumvented by
simulating a forbidden operation through a sequence of allowed
operations.  Fundulaki and Maneth~\cite{fundulaki07sacmat} introduced
this problem and showed that it is undecidable for full XPath.
Moore~\cite{moore11ic} further investigated the complexity of special
cases of this problem. Bravo et al.~\cite{bravo12vldbj} studied
schema-based policies for which consistency is in \PTIME\ and also
investigated repair algorithms for inconsistent policies.  Jacquemard
and Rusinowitch~\cite{jacquemard10ppdp} studied complexity and
algorithms for consistency of policies with respect to richer classes
of schemas.  Fairness and consistency are orthogonal concerns.

Language-based security, particularly analysis of information flow, is
another security problem that has been studied
extensively~\cite{sabelfeld03sac}, including for XML
transformations~\cite{benzaken03asian}.  This paper considers only
classical access control (deciding whether to allow or deny actions
specified by a policy), a largely separate concern.  Thus, while our
approach draws on ideas familiar from language-based security such as
static analysis, the key problems for us are different.  Typically,
language-based information flow security aims to provide a
conservative upper bound on possible run-time behaviors of programs,
for example to provide a non-interference guarantee. Thus, sound
over-approximation is tolerable for information-flow security.  In
contrast, we wish to exactly enforce fine-grained access control
policies, so we need to consider exact static analyses and related
properties such as fairness.

\section{Conclusion}
\label{sec:concl}

Fine-grained, rule-based access control policies for XML data are
expensive to enforce by dynamically checking whether the update
complies with the rules.  In this paper, we advocate enforcement based
on static analysis, which is equivalent to dynamic enforcement when
the policy is fair.  We gave a novel topological characterization of
fairness, and used this
characterization to prove that for policies over $\XPCDSF$, all
policies without filters in negative rules are fair (with respect to
$\XPCF$, and fairness is decidable in \coNP-time.  

There are natural next steps for future work, including investigating
fairness for larger fragments of XPath or in the presence of schemas
or constraints on the data, and generalizing the approach to ordered
trees and the full complement of XPath axes.  Implementing and
evaluating the practicality of fair policy enforcement or fairness
checking is also of interest.  Finally, our approach places the burden
of finding an appropriate statically allowed $U$ that covers a desired
update $u$ on the user; it may be necessary to develop efficient
techniques for automating this process.



\paragraph{Acknowledgements} Thanks to Irini Fundulaki and Sebastian
Maneth for discussions on XPath access control policies.


\newpage
\appendix
\section{Proofs}

\begin{proof}[Proof of Lem.~\ref{lem:intersection}]
  Suppose $(V',k) \in \AB{T,n} \cap \AB{U,m}$.  Then there must exist witnessing
  homomorphisms $ h_1 : (T,n) \to (V,k)$ and $h_2 : (U,m) \to (V,k)$.
  Without loss of generality, assume that $h_1$ and $h_2$ are
  injective and $rng(h_1) \cap rng(h_2)$
  consists only of the vertices between $R_V$ and $k$.  Observe that
  $h_1$ and $h_2$ are invertible when restricted to $rng(h_1) \cap
  rng(h_2)$.

  Construct $(V,k)$ from $(V',k)$ by deleting all subtrees that do not
  contain a node from $rng(h_1) \cup rng(h_2)$.  Observe that this
  implies that $V_V = rng(h_1) \cup rng(h_2)$.
To see that $\AB{V,k}
  = \AB{T,n} \cap \AB{U,m}$, the forward inclusion $\AB{V,k} \subseteq
  \AB{T,n} \cap \AB{U,m}$ is immediate.  Suppose $(W,l)
  \in \AB{T,n} \cup \AB{U,m}$, and suppose $h_1' : (T,n) \to (W,l)$
  and $h_2' : (U,m) \to (W,l)$ are homomorphisms witnessing this.  
  Choose a function $g : (V,k) \to (W,l)$ such that: 
\[g(x) = \left\{
  \begin{array}{ll}
    h_1'(y) & x \in rng(h_1), h_1(y) = x\\
    h_2'(z) & x \in rng(h_2) - rng(h_1), h_2(z) = x 
  \end{array}\right.
\]
We first show that $h_1' = g \circ h_1$ and $h_2' = g \circ h_2$. The
first equation is immediate; for the second, clearly $g(x) = h_2'(x)$
when $x \in rng(h_2) - rng(h_1)$.  If $x \in rng(h_1) \cap rng(h_1)$
then $x$ is between $R_V$ and $k$, so $h_1^{-1}(x) = \{y\}$ and
$h_2^{-1}(x) = \{z\}$ where $y$ and $z$ are in the corresponding
position on the paths between $R_T$ and $n$ and $R_U$ and $m$
respectively.  Thus, we must have that $g(z) = h_1'(y) = h_2'(z)$
because both $h_1'$ and $h_2'$ are homomorphisms.  

To show that $g$ is
a homomorphism, first $g(R_V) = h_1'(R_T) = R_W$.  Second, for any
edge $(v,w) \in E_V$, there are several cases to show that
$(g(v),g(w)) \in E_W$.  If $w \in rng(h_1)$ then clearly $v \in
rng(h_1)$ also, and $v = h_1(v'), w=h_1(w')$ where $(v,w) \in E_T$ by the injectivity
of $h_1$, so 
then $(g(v),g(w)) = (g(h_1(v')), g(h_1(w'))) = (h_1'(v'), h_1'(w'))
\in E_W$.  Similarly, if $w \in rng(h_2)$ we are done.  
Finally, for any $v \in V_V$, there are several cases to consider in
showing $\lambda_W(g(v)) = \lambda_V(v)$.  If $v \in
\rng(h_1)$ then suppose $v = h_1(v')$ for some $v' \in V_T$.  Then $\lambda_W(g(v)) =
\lambda_W(g(h_1(v'))) = \lambda_W(h_1'(v')) = \lambda_T(v') =
\lambda_V(h_1(v')) = \lambda_V(v)$.  The case for $v \in rng(h_2)$ is similar.
\end{proof}







\end{document}